\documentclass[sigconf,9pt]{acmart}

\usepackage{float}
\usepackage{amsmath}
\usepackage{caption}
\usepackage{subcaption}

\begin{document}
\title{Market Model for Demand Response under Block Rate Pricing}

%
\author{Haris Mansoor}
\affiliation{%
	\institution{Department of Computer Science, \\
		Lahore University of Management Sciences}
	\city{Lahore}
	\state{Pakistan}
}
\email{16060061@lums.edu.pk}

\author{Naveed Arshad}
\affiliation{%
 \institution{Department of Computer Science, \\
 	Lahore University of Management Sciences}
 \city{Lahore}
 \state{Pakistan}
}
\email{naveedarshad@lums.edu.pk}


\renewcommand{\shortauthors}{A et al.}

\begin{abstract}



Renewable sources are taking center stage in electricity generation. However, matching supply with demand in a renewable-rich system is a difficult task due to the intermittent nature of renewable resources (wind, solar, etc.). As a result, Demand Response (DR) programs are an essential part of the modern grid. An efficient DR technique is to devise different pricing schemes that encourage customers to reduce or shift the electric load.

In this paper, we consider a market model for DR using Block Rate Pricing (BRP) of two blocks. We use a utility maximization approach in a competitive market. We show that when customers are price taking and the utility cost function is quadratic the resulting system achieves an equilibrium. Moreover, the equilibrium is unique and efficient, which maximizes social welfare. A distributed algorithm is proposed to find the optimal pricing of both blocks and the load. Both the customers and the utility runs the market. The proposed scheme encourages customers to curtail or shift their load. Numerical results are presented to validate our technique.

\end{abstract}

%
%



\keywords{Smart grid, Demand Response, Real time pricing, MArket Model, Block Rate Pricing, Utility maximization, Social Welfare, Distribute Optimization}

\maketitle

\section{Introduction}

Many countries have plans to shift to renewable sources of electricity by 2050 \cite{inproceedings1,article}. This is mainly due the environmental and cost related problems with the fossil fuel based power plants. While renewable sources (such as solar and wind) provide an environment friendly solution to the energy demand, they also create challenges in the electricity distribution system. The problems mostly originate from the variable and stochastic nature of renewable sources which produces a mismatch between the demand and supply of electricity \cite{lusis2017short}.

With the current penetration level of renewable resources, utilities use different approaches to handle the demand and supply gap. Most commonly used methods include fossil fuel based peaker power plants
\cite{vardakas2015survey} and  
 Demand side management (DSM) techniques such as  
Demand Response (DR), spinning reserve, electricity storage  \cite{roberts2011role,mohd2008challenges} and 
 electricity curtailment \cite{aalami2010demand}.

Electricity storage methods and peaker power plants are very costly to operate, while electricity curtailment has negative impact on the life style of customers. DR programs are the most effective way to match demand with supply. Based on the techniques used to encourage customers to reduce or shift their load, DR programs are broadly divided into two categories \cite{vardakas2015survey}. $(1)$ Price based programs use different pricing schemes $(2)$ Incentive based programs use financial or other benefit to enroll customers in the DR programs.

The price based DR programs consist of
\begin{enumerate}
	\item{{\em Flat Pricing}: Fixed price per unit of electricity  }
	\item{{\em Block Rate Pricing (BRP)}: Flat price but changes under different usage (blocks) of electricity}
	\item{{\em Time of Use Pricing (TOU)}: Flat price during different time periods of the day }
	\item{{\em Critical Peak Pricing (CPP)}: Price is not flat and can changes under different system conditions}
		\item{{\em Real Time Pricing (RTP)}: Price changes in each time intervals}

\end{enumerate}

In this paper, we use BRP as a DR program to motivate customers to reduce or shift their load. In BRP there is a threshold of electricity, the price per unit increases or decreases as electricity consumption increases beyond the threshold (blocks). This can discourage or encourage customers to use less or more electricity \cite{cardenas2019consequences}. We consider a situation in which the power supply is flexible and there are two blocks of BRP. The consumers maximize their net utility over the whole day. It is proved that under BRP and quadratic cost function the system achieve equilibrium that maximizes social welfare. A distributed algorithm is proposed to find the equilibrium. The implementation of proposed distributed algorithm requires two way communication between the utility company and the customers. Such kind of communication is possible with Advanced Metering Infrastructure (AMI).

The paper is organized as follow. Literature review is presented in section \ref{relatedwork}. In section \ref{Proposed_methodology} utility maximization problem is formulated under BRP. Distributed algorithm and numerical results are presented in section \ref{Experiments}. Conclusions  are drawn in  section \ref{conclusions}.

\section{Related Work}\label{relatedwork}


A large number of related literature exists on DR and different market models for the power system. In this section, we discuss the work that is related to our approach.

In \cite{chen2010two} two market models for demand matching and demand shaping are presented. The equilibrium of these models is studied under competitive and oligopolistic situations and distributed algorithms are presented for their solution. Similarly, in \cite{vidyamani2019demand} a demand response solution based on utility maximization using TOU pricing is derived. The effect of TOU pricing in power system with plug-in hybrid electric vehicle is studied in \cite{shao2010impact} and it is concluded that TOU price is an effective method to reduce the peak demand. An extensive study of different pricing methods on residential users is presented in \cite{newsham2010effect}. The data suggests the CPP along with automatic curtailment is a good DR technique. However, it can cause serious hardships for consumers. A real time pricing model under many distributed generation units is proposed in \cite{huang2013dynamic}. A distributed algorithm is also presented that can solve the problem. The effect of BRP on residential customers in chine is studied in \cite{lin2012designation} and some interesting findings are reported. The BRP will substantially improve the equity and efficiency of residential customers. Secondly the BRP the electricity consumption of residential customers will be reduced, which results in the reduction of carbon emission.

\section{Problem Formulation} \label{Proposed_methodology}

Consider a power system with a set $ N $ of users/customers and a set $ T $ of time slots in a day. These customers are served by one power company. For each customer $ i \in N $ and $t \in T $ there is a power consumption denoted by $ x_{i}^{t} $. Moreover, each customer has a minimum and maximum power consumption in a day. The minimum power consumption corresponds to basic requirements over the whole day and maximum power consumption corresponds to all equipment running the whole day. 
\begin{equation}\label{eq1}
	\sum_{t=1}^{T} x_{i}^{t} \geq D_{i}^{min},\ \  i \in N 
\end{equation}
\begin{equation}\label{eq2}
	\sum_{t=1}^{T} x_{i}^{t} \leq D_{i}^{max},\ \  i \in N 
\end{equation}

We consider BRP with two blocks, the pricing of first and the second block is represented by $P_{l}^{t}$ and $P_{l}^{t}$ respectively. 
For every time $t$ there is a predefined threshold $b^t$. When electricity consumption is increased beyond the threshold $ b^t $ the electricity price increases or decreases step wise.
\begin{equation} \label{eq3}
	P^{t} =
	\begin{cases}
		P_{l}^{t}       & \quad \text{if } 0\leq x_{i}^{t} \leq b^t \\
		P_{u}^{t}   & \quad \text{if } b^t\le x_{i}^{t}
	\end{cases}
\end{equation}
After consuming $ x_{i}^{t} $ power each user obtains a utility $ U( x_{i}^{t}) $, which represents the level of satisfaction of a user as a function of power. The utility function should be non-decreasing and concave such that zero utility corresponds to zeros power consumption \cite{chen2010two}. We can write $x_{i}^{t}$ into two parts such that one is equal to and less than $b^t$ and the other is greater than or equal to $b^t$.
\begin{equation}\label{eq4}
	x_{i}^{t}=\text{min}(x_{i}^{t},b^t)+\text{max}(x_{i}^{t},b^t)-b^t
\end{equation}
The power company incurs a cost $ C(D,t) $ for providing $ D $ demand at time $ t $. The cost function is ideal to model time dependent power production such as renewable energy resources. We assume that the cost function is quadratic in $D$.

The goal of power company is to maximize its net revenue. Which is profit minus cost.The profit and cost for a single time slot $t$ are given below.
\begin{equation}\label{eq5}
	Profit= P_{l}^{t}*\sum_{i \in N}\text{min}(x_{i}^{t},b^t)+
	P_{u}^{t}*\sum_{i \in N}(\text{max}(x_{i}^{t},b^t)-b^t)
\end{equation}
\begin{equation}\label{eq6}
	Cost=C \big(\sum_{i \in N}\text{min}(x_{i}^{t},b^t)+\sum_{i \in N}(\text{max}(x_{i}^{t},b^t)-b^t),t \big)
\end{equation}
We can replace min and max functions with dummy variables by adding additional constraints \cite{boyd2004convex}.
\begin{equation}\label{eq7}
	\begin{aligned}
		& \text{min}(x_{i}^{t},b^t)=y_{i}^{t} \\
		& \text{subject to:}\ \ y_{i}^{t} \leq x_{i}^{t};\ y_{i}^{t} \leq b^t
	\end{aligned}    
\end{equation}

\begin{equation}\label{eq8}
	\begin{aligned}
		& \text{max}(x_{i}^{t},b^t)=z_{i}^{t} \\
		& \text{subject to:}\ \ z_{i}^{t} \geq x_{i}^{t};\ z_{i}^{t} \geq b^t
	\end{aligned}    
\end{equation}
By using the above variables we can write $x_{i}^{t}$ as
\begin{equation}\label{eq9}
	x_{i}^{t}=y_{i}^{t}+z_{i}^{t}-b^t;\ \ \ \ i \in N, t \in T
\end{equation}
Using $y_{i}^{t}$ and $z_{i}^{t}$ the above net revenue problem becomes.
\begin{equation}\label{eq10}
	\begin{aligned}
		& \underset{\substack{y_{i}^{t} \leq x_{i}^{t} \\ y_{i}^{t} \leq b^t \\ z_{i}^{t} \geq x_{i}^{t} \\ z_{i}^{t} \geq b^t \\ x_{i}^{t} \geq 0 \\ x_{i}^{t}=y_{i}^{t}+z_{i}^{t}-b^t } }{\text{Maximize:}}
		\sum_{t \in T} \bigg( P_{l}^{t}\sum_{i \in N}y_{i}^{t}+ P_{u}^{t}\sum_{i \in N}(z_{i}^{t}-b^t)-C(\sum_{i \in N}(y_{i}^{t}+z_{i}^{t}-b^t),t) \bigg) \\
	\end{aligned}
\end{equation}
By taking partial derivatives the solution of above problem becomes.
\begin{equation}\label{eq11}
	P_{l}^{t}=\frac{\partial C(\sum_{i \in N}(y_{i}^{t}+z_{i}^{t}-b^t),t)}{\partial (\sum_{i \in N} y_{i}^{t})}; \ \ \   t \in T
\end{equation}
\begin{equation}\label{eq12}
	P_{u}^{t}=\frac{\partial C(\sum_{i \in N}(y_{i}^{t}+z_{i}^{t}-b^t),t)}{\partial (\sum_{i \in N} z_{i}^{t})}; \ \ \   t \in T
\end{equation}
\subsection{Utility Maximization}
In this section we consider a competitive market where customers take the price communicated by power company. For real time block rate pricing $P_{l}^{t}$ and $P_{u}^{t}$ the customer $i$ allocates its energy usage to maximize its aggregated net utility,along with constraints \eqref{eq1} and \eqref{eq2}.
\begin{equation}\label{eq13}
	\begin{aligned}
		& \underset{\substack{ x_{i}^{t} \geq 0 } }{\text{Maximize:}}
		& & \sum_{t \in T} \bigg( U(x_{i}^{t})-P_{l}^{t}\text{min}(x_{i}^{t},b^t)+P_{u}^{t}(\text{max}(x_{i}^{t},b^t)-b^t) \bigg) \\
		& \text{Subject to:}
		& & \sum_{t=1}^{T} x_{i}^{t} \geq D_{i}^{min}
		\\
		& & &\sum_{t=1}^{T} x_{i}^{t} \leq D_{i}^{max}
	\end{aligned}
\end{equation}
Using the \eqref{eq7}, \eqref{eq8} and \eqref{eq9} the above problem can be transformed into 
\begin{equation}\label{eq14}
	\begin{aligned}
		& \underset{\substack{y_{i}^{t} \leq x_{i}^{t} \\ y_{i}^{t} \leq b^t \\ z_{i}^{t} \geq x_{i}^{t} \\ z_{i}^{t} \geq b^t \\ x_{i}^{t} \geq 0 \\ x_{i}^{t}=y_{i}^{t}+z_{i}^{t}-b^t } }{\text{Maximize:}}
		& & \sum_{t \in T} \bigg( U(y_{i}^{t}+z_{i}^{t}-b^t)-P_{l}^{t}y_{i}^{t}- P_{u}^{t}(z_{i}^{t}-b^t) \bigg) \\
		& \text{Subject to:}
		& & \sum_{t=1}^{T} (y_{i}^{t}+z_{i}^{t}-b^t) \geq D_{i}^{min}
		\\
		& & &\sum_{t=1}^{T} (y_{i}^{t}+z_{i}^{t}-b^t) \leq D_{i}^{max}
	\end{aligned}
\end{equation}
The above model captures both the real time block rate pricing and demand shifting. The sum over the $t$ enforce the demand shifting.
The solution of above problem is the equilibrium of the competitive market problem \eqref{eq14}. We can find the equilibrium using the Lagrange dual function by using Lagrange multipliers $\lambda_{i}^1 \geq 0$ and $\lambda_{i}^2 \geq 0$ for constraints in \eqref{eq14} . The optimal $y_{i}^{t}$ and $z_{i}^{t}$ are given below by following KKT conditions \cite{boyd2004convex}.
\begin{equation}\label{eq15}
	\begin{aligned}
		& \frac{\partial U(y_{i}^{t}+z_{i}^{t}-b^t,t)}{\partial y_{i}^{t}}-P_{l}^{t}-\lambda_{i}^1+\lambda_{i}^2=0; \ \ \  i \in N, t \in T \\
		& \frac{\partial U(y_{i}^{t}+z_{i}^{t}-b^t,t)}{\partial z_{i}^{t}}-P_{u}^{t}-\lambda_{i}^1+\lambda_{i}^2=0; \ \ \  i \in N, t \in T \\
		& \lambda_{i}^1(\sum_{t=1}^{T} (y_{i}^{t}+z_{i}^{t}-b^t) - D_{i}^{max}=0); \ \ \  i \in N \\
		& \lambda_{i}^2(D_{i}^{min}-\sum_{t=1}^{T} (y_{i}^{t}+z_{i}^{t}-b^t)=0); \ \ \ i \in N
	\end{aligned}  
\end{equation}
A competitive equilibrium for demand response system is a system of equations that simultaneously solves \eqref{eq11}, \eqref{eq12} and \eqref{eq15}. The equilibrium consists of variables $\{y_{i}^{t},z_{i}^{t}, P_{l}^{t},P_{u}^{t}\}$ that solves equations \eqref{eq11}, \eqref{eq12} and \eqref{eq15} simultaneously.
\begin{theorem}
	There is a unique equilibrium for a competitive demand response system under block rate pricing. Moreover, the equilibrium maximizes the social welfare.
	\begin{equation}\label{eq16}
		\begin{aligned}
			& \underset{\substack{y_{i}^{t} \leq x_{i}^{t} \\ y_{i}^{t} \leq b^t \\ z_{i}^{t} \geq x_{i}^{t} \\ z_{i}^{t} \geq b^t \\ x_{i}^{t} \geq 0 \\ x_{i}^{t}=y_{i}^{t}+z_{i}^{t}-b^t } }{\text{Maximize:}}
			& & \sum_{t \in T} \bigg( \sum_{i \in N}U(y_{i}^{t}+z_{i}^{t}-b^t)-C(\sum_{i \in N}(y_{i}^{t}+z_{i}^{t}-b^t),t) \bigg) \\
			& \text{Subject to:}
			& & \sum_{t=1}^{T} (y_{i}^{t}+z_{i}^{t}-b^t) \geq D_{i}^{min},\ \  i \in N \\
			& & &\sum_{t=1}^{T} (y_{i}^{t}+z_{i}^{t}-b^t) \leq D_{i}^{max},\ \  i \in N 
		\end{aligned}
	\end{equation}
\end{theorem}
\begin{proof}
	We have assumed that our cost function is quadratic, under this assumption $ \frac{\partial C(\sum_{i \in N}(y_{i}^{t}+z_{i}^{t}-b^t))}{\partial (\sum_{i \in N} y_{i}^{t})} = \frac{\partial C(\sum_{i \in N}(y_{i}^{t}+z_{i}^{t}-b^t))}{\partial y_{i}^{t}}$. Using this fact if we put \ref{eq11} and \ref{eq12} into \ref{eq15}, we get the equilibrium condition for the social welfare problem \ref{eq16}. Moreover the equilibrium is unique since the problem \ref{eq16} is  strictly convex.
	\begin{equation}\label{eq17}
		\begin{aligned}
			& \frac{\partial U(y_{i}^{t}+z_{i}^{t}-b^t)}{\partial y_{i}^{t}}-\frac{\partial C(\sum_{i \in N}(y_{i}^{t}+z_{i}^{t}-b^t))}{\partial y_{i}^{t}}-\lambda_{i}^1+\lambda_{i}^2=0; \ \ \  \\
			& \frac{\partial U(y_{i}^{t}+z_{i}^{t}-b^t)}{\partial z_{i}^{t}}-\frac{\partial C(\sum_{i \in N}(y_{i}^{t}+z_{i}^{t}-b^t))}{\partial z_{i}^{t}}-\lambda_{i}^1+\lambda_{i}^2=0; \ \ \  \\
			& \lambda_{i}^1(\sum_{t=1}^{T} (y_{i}^{t}+z_{i}^{t}-b^t) - D_{i}^{max}=0); \ \ \ i \in N\\
			& \lambda_{i}^2(D_{i}^{min}-\sum_{t=1}^{T} (y_{i}^{t}+z_{i}^{t}-b^t)=0); \ \ \ i \in N
		\end{aligned}  
	\end{equation}
\end{proof}




%







\section{Distributed Algorithm and Numerical Results}\label{Experiments}
In this section we first propose a distributed algorithm to find the equilibrium and then we present the results that show the convergence of algorithm to optimal pricing and customers demand.
\subsection{Distributed Algorithm}
The social welfare optimization problem \ref{eq16} can be solved in a distributed and iterative way using gradient based algorithm \cite{bertsekas1989parallel}. We assume that the bracket interval $b^t$ are already defined and communicated to the users. 
The whole procedure repeats for every time step.

At the start of $k^{th}$ iteration

\begin{itemize}

	\item {The utility company communicates to each user and receives the demand data  $(x_{i}^{t+1})^k$ for the next time slot $t+1$. Using $b^{t+1}$ and $(x_{i}^{t+1})^k$, the utility company computes $(y_{i}^{t+1})^k,(z_{i}^{t+1})^k$ }
	\item{Using $(y_{i}^{t+1})^k,(z_{i}^{t+1})^k$ the utility company computes prices for both brackets and communicate them to users
	\begin{equation}\label{eq18}
	(P_{l}^{t+1})^{k}=\frac{\partial C(\sum_{i \in N}((y_{i}^{t+1})^{k}+(z_{i}^{t+1})^{k}-b^{t+1}),t)}{\partial (\sum_{i \in N} (y_{i}^{t+1})^{k})}
\end{equation}
\begin{equation}\label{eq19}
	(P_{u}^{t+1})^{k}=\frac{\partial C(\sum_{i \in N}((y_{i}^{t+1})^{k}+(z_{i}^{t+1})^{k}-b^{t+1}),t)}{\partial (\sum_{i \in N}  (z_{i}^{t+1})^{k})}
\end{equation} }
	\item{After receiving the price each user $i$ update its demand using gradient based method \cite{bertsekas1989parallel}.
	
	\begin{equation}\label{eq20}
	(y_{i}^{t+1})^{k+1}=[(y_{i}^{t+1})^{k}+\gamma(\frac{\partial U((y_{i}^{t+1})^{k}+(z_{i}^{t+1})^{k}-b^{t+1}}{\partial (y_{i}^{t+1})^{k}}) - (P_{l}^{t+1})^{k})] ^{s_{i}}
\end{equation}
	
	\begin{equation}\label{eq21}
	(z_{i}^{t+1})^{k+1}=[(z_{i}^{t+1})^{k}+\gamma(\frac{\partial U((y_{i}^{t+1})^{k}+(z_{i}^{t+1})^{k}-b^{t+1}}{\partial (z_{i}^{t+1})^{k}}) - (P_{u}^{t+1})^{k})] ^{s_{i}}
\end{equation}
	
	}
\end{itemize}

\begin{figure*}[h]
	\centering
	\begin{subfigure}{.33\textwidth}
		\includegraphics[width=5cm,height=3cm]{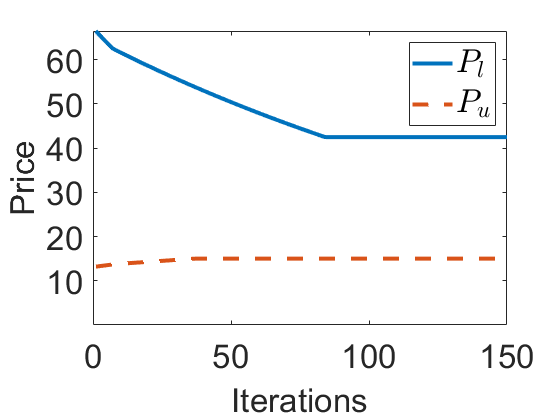}
		\caption{}
	\end{subfigure}
	\begin{subfigure}{.33\textwidth}
		\includegraphics[width=5cm,height=3cm]{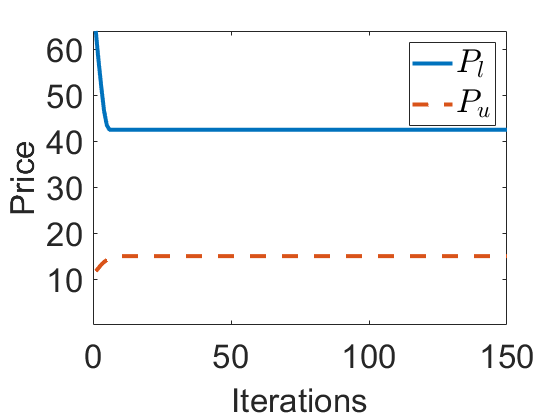}
		\caption{}
	\end{subfigure}
	\begin{subfigure}{.33\textwidth}
		\includegraphics[width=5cm,height=3cm]{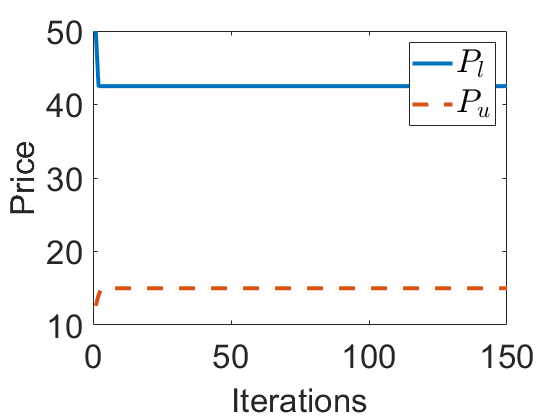}
		\caption{}
	\end{subfigure}
	
	\caption{Effect of $\gamma$ on the convergence of algorithm for $\beta_2 >\beta_1$ (a) $\gamma=0.01$ (b) $\gamma=0.1$ (c) $\gamma=0.3$}
	\label{fig:alpha_0_1}
\end{figure*}

%

Where $\gamma > 0$ is a constant step size used in gradient methods \cite{bertsekas1989parallel}. The set $s_i$ represents the projection on the constraints given by \ref{eq1},\ref{eq2}, \ref{eq7}, \ref{eq8} and \ref{eq9}. This operation can be easily performed by the users since all these constraints are of same user. Since the problem is convex, it converges for small $\gamma$. The whole algorithm repeats for each time step $t$. Both the power company and users jointly run the market under BRP.  

\subsection{Numerical Example}

We consider a simple numerical example of two customers($N=2$), with $b^t=25$. Thw utility function used is given in eq \ref{eq22}, where $w_i \in [10-100]$ and $\alpha =1$. We consider a step wise quadratic cost function, Where the cost changes after providing first $b^t*N$ units of electricity. If the cost increases it discourages customers to use less electricity and $P_{u}^{t}>P_{l}^{t}$. Similarly, if the cost of production decreases it encourages customers to use more electricity and $P_{l}^{t}>P_{u}^{t}$. This cost function is more generic since there is a base power fulfilled by one set of generation units and as the demand increases the utility runs additional generation units with different cost. The cost function is given in eq \ref{eq23}.

\begin{equation} \label{eq22}
	U(x^{t}_{i}) =
	\begin{cases}
		w_i x^{t}_{i}-\alpha  (x^{t}_{i})^2      & \quad \text{if } 0\leq x_{i}^{t} \leq \frac{w_i}{\alpha} \\
		\frac{w_{i}^{2}}{2\alpha}   & \quad \text{if } \frac{w_i}{\alpha}\le x_{i}^{t}
	\end{cases}
\end{equation}

\begin{equation} \label{eq23}
	C(D) =
	\begin{cases}
		\beta_{1}D^2      & \quad \text{if } 0\leq D \leq b^tN \\
		\beta_{2}D^2   & \quad \text{if }  D > b^tN
	\end{cases}
\end{equation}

Figure \ref{fig:alpha_0_1} shows the effect of different $\gamma$ on the convergence of algorithm, with  $\beta_2 >\beta_1$. As $\gamma$ increases the algorithm started to converge in less iterations. Since  $\beta_2 >\beta_1$ the price of second block is greater than the price of first block ($P_u>P_l$). 

\section{Conclusion} \label{conclusions}
In this paper we formulate a DR program under block rate pricing with two blocks. We consider a competitive market where users try to maximize their net utility and the power company try to maximize the net revenue. It is proved that under block rate pricing and quadratic cost function, the system achieves an equilibrium, which is unique and efficient. We propose a distributed algorithm to compute the equilibrium. The numerical results demonstrate that the distributed algorithm can efficiently compute the pricing of both blocks.

\bibliographystyle{ACM-Reference-Format}
\bibliography{reference}

\end{document}